\documentclass[conference]{IEEEtran}

\usepackage{blindtext}
\usepackage[utf8]{inputenc}

\usepackage{pdfpages}
\usepackage{graphicx}
\graphicspath{ {/Users/chedia/Downloads/} }
\usepackage{epstopdf}
\usepackage{amsthm}
\usepackage{amsmath, amssymb}
\usepackage{bbm}

\title{internship report}

\newcommand{\reels}{\mathbb{R}}
\newcommand{\esperance}{\mathbb{E}}

\newtheorem{Pro}{Proposition}

\newtheorem{Cor}{Corollary}

\newtheorem*{proof*}{Proof}

\begin{document}

\title{The Effects of Mobility on the Hit Performance of Cached D2D Networks}

\author{Chedia~Jarray$^\dagger$ and Anastasios~Giovanidis$^\ast$\\[2ex]}
\maketitle

\begin{abstract}
%\boldmath
A device-to-device (D2D) wireless network is considered, where user devices also have the ability to cache content. In such networks, users are mobile and communication links can be spontaneously activated and dropped depending on the users' relative position. Receivers request files from transmitters, these files having a certain popularity and file-size distribution. In this work a new performance metric is introduced, namely the Service Success Probability, which captures the specificities of D2D networks. For the Poisson Point Process case for node distribution and the $\mathrm{SNR}$ coverage model, explicit expressions are derived. Simulations support the analytical results and explain the influence of mobility and file-size distribution on the system performance, while providing intuition on how to appropriately cache content on mobile storage space. Of particular interest is the investigation on how different file-size distributions (Exponential, Uniform, or Heavy-Tailed) influence the performance.
\end{abstract}


\let\thefootnote\relax\footnotetext{\hspace{-2ex}$^\dagger$UR MACS, ENIG-University of Gabes,  
Rue Omar Ibn Elkhattab, 6029 Gabes, Tunisia; jarray.chedia@hotmail.fr,\\
The doctoral visit of Chedia Jarray at T\'el\'ecom ParisTech was funded as a scholarship by the Ministry of Higher Education of Tunisia.
\\
$^\ast$CNRS-LTCI and T\'el\'ecom ParisTech,
23 avenue d'Italie,
75013 Paris, France; anastasios.giovanidis@telecom-paristech.fr}
\newcommand{\thefootnote}{\arabic{footnote}}

\section{Introduction}
The relatively recent commercial spread of new generation mobile devices such as tablets and smart-phones has triggered an explosive increase of data traffic and has made traffic management crucial for communication networks. Currently, mobile video streaming accounts for almost half of the mobile data traffic and is expected to have a considerable increase over the next years. These developments compel mobile operators to redesign their current networks and seek more advanced techniques to increase coverage, boost network efficiency, and cost-effectively bring content closer to the user, either by deploying small base-stations (BSs) \cite{golrezaei2013femtocaching}, \cite{bastug2014living}, or by exploring the possibilities for inter-device communication, known as D2D (device-to-device)  \cite{golrezaei2013femtocaching}, \cite{ADhill15b}.

We are particularly interested here in the potential of D2D communications, where mobile devices also play role in content delivery, and direct communication links between users are enabled. Such solution will possibly exist on top of the existing cellular infrastructure and is already envisioned for 5G networks. Furthermore, we consider the possibility of on-device content caching. The motivation is that, among existing (multimedia) content, only a small fraction is repeatedly used, which however triggers the majority of the total data traffic. Motivated by this, and the fact that mobile devices are equipped with cheap and relatively large storage capacity, caching finite popular files on mobile devices in advance could be promising to relieve the overloaded network traffic \cite{molisch2014caching}, \cite{ji2013fundamental}. The idea of D2D caching can significantly offload different parts of the network including the radio access network, core network, and backhaul, by smartly prefetching and storing contents on the user nodes. Because of the D2D links, multimedia files can be transmitted from one user to another with reduced latency. 

An important aspect in D2D communication that makes its design challenging is \textit{user mobility}. The fact that user nodes constantly change relative position is one of the major factors that can diminish possible benefits of this type of data transmission. \textit{A partly transmitted file due to connection loss can be completely useless}. 

In this work the performance of caching in D2D networks is studied, for different degrees of node mobility. Specifically, a D2D network is considered, where \textit{devices} (mobile nodes) are spatially distributed on the plane. The possibility of communication between a mobile node and a fixed station is left out, in this scenario. At some point in time a subset of these devices (\textit{receivers}) requests for data files, whereas the other nodes can serve them as potential \textit{transmitters}. Based on channel quality, reliable wireless links can be established between receivers and transmitters, that satisfy a required Quality-of-Service (QoS). If some of these transmitters also have the desired content cached, then transmission is initialised (at most one transmitter per receiver). However, nodes change position over time and consequently the link quality between receivers and transmitters is affected. An established link with sufficient quality at $t_o$, can be later dropped at some $t_1>t_o$ due to displacement of one of the two nodes in pair, and the consequent quality degradation.  In our model, connection loss due to transmitter displacement is considered an unsuccessful effort. The node mobility inherent in the nature of D2D communications is what this work wishes to study and provide a method to analyse performance metrics of interest.

Our contributions are the following:

\begin{itemize}
\item A mobility concept is introduced, where a node keeps its position for an exponential amount of time before being displaced far from the receiver. 
\item Content \textit{does not} have the same size, rather file-sizes take values sampled from relevant probability distributions. Possible such distributions are the Exponential, Uniform or a Heavy-Tailed one. 
\item New performance metrics are introduced: the Total and Expected Service Success Probability. For these, explicit expressions are provided for the case when node positions follow a Poisson Point Process (PPP). and for the $\mathrm{SNR}$ connectivity model. 
\item Performance evaluation is provided through comparison plots, both from simulations and analysis. Conclusions are drawn over the influence of mobility and file-size distribution on cached D2D performance. 
\end{itemize}

\subsection{Literature}

In the literature, a significant amount of work analyses content caching in wireless networks. Specifically, FemtoCaching is proposed by Shanmugam et al in \cite{shanmugam2013femtocaching} where caching helpers optimally store popular content for delay performance improvement. Bastug et al in \cite{bastug2014living}, treat the problem of proactive caching by use of stochastic geometry, and show the gains in terms of backhaul savings and user satisfaction. B{\l}aszczyszyn and Giovanidis in \cite{BlaGioICC15} study the optimal probabilistic placement policy for maximizing the total hit probability in random network topologies. Molisch et al \cite{molisch2014caching} evaluate the performance of caching on helper station/devices and optimise video quality by proposing optimal storage schemes. Content placement for delay-tolerant service satisfaction is analysed by Sermpezis et al in \cite{SeSpy15}. Asymptotic laws of the required link capacity of a cached multi-hop wireless network are studied by Paschos et al in \cite{PaschWIOPT12}. Related to the problem of routing and replica placement in a network, Sourlas et al propose various on-line autonomous cache management algorithms \cite{SourlasAlgoN13}.

Cached D2D communications is treated by Ji et al in \cite{ji2013fundamental}, where the authors find asymptotic throughput scaling laws with coded caching and D2D spatial reuse. Afshang et al analyse cached coverage in a clustered PPP model in \cite{ADhill15b}. Mobility in cellular networks is an important topic and interesting analytical models have been proposed by Lin et al \cite{GantiMob13}, and Hsu et al \cite{Hsu09}. The optimal storage allocation when user mobility is modelled by a a Markov chain random walk is approximately solved by Poularakis and Tassiulas in \cite{poularakis2013exploiting}. To optimally store content, the authors in \cite{ZhuoCont11} formulate and solve a contact-duration-aware data replication optimisation problem.

The remainder of this paper is organized as follows. Section \ref{secII} describes in detail the system model, whereas the relevant performance metrics are introduced in Section \ref{secIII}. Section \ref{secIV} contains the main analytical results of the stochastic geometry analysis, followed by certain special cases of interest and a discussion. Performance evaluation of the model is provided in Section \ref{secV}, where simulations confirm the validity of the derived analytical expressions. Furthermore, plots illustrate how the performance is influenced by system parameters, such as the mean node lifespan or the file-size order and distribution. Conclusions of the work are drawn in Section \ref{secVI}.

 %--------------------------------------------------------------------------------------------

\section{System Model}
\label{secII}

\subsection{Physical Aspects}

The transmitters and receivers are placed following two independent planar homogeneous PPPs $\Phi_t$ and $\Phi_r$ with intensity $\lambda_t>0$ and $\lambda_r>0$, respectively. Their superposition is also a PPP with sum intensity $\lambda=\lambda_t+\lambda_r$ [$devices/m^2$]. The transmitter devices are enumerated $i=1,2,\ldots,$ and we denote by \textit{$x_i$} the location of the \textit{$i^{th}$} transmitter. Since the receivers follow a PPP as well, we can condition on a typical receiver located at the Cartesian origin $o=(0,0)$. The Euclidian distance between the origin and each transmitter at $x_i$ is denoted by $r_i:=|x_i|$. According to the Slivnyak-Mecke theorem and the stationarity and isotropy of the PPP, the results for the typical receiver are valid for any receiver of $\Phi_r$ randomly located on the 2D plane \cite{BacBlaVol1}.

For the signal propagation model, the path-loss from transmitter to receiver is equal to $l(r_i)=r_{i}^{-\alpha}$, where $\alpha>2$. All transmitters are assumed to emit the same power level equal to P [$Watt$]. Let $H_i$ be the random variable for channel fading between $x_i$ and $o$, with unit average. These variables indexed by $i$ are independent and identically distributed and the generic fading variable is denoted by $H$ (no specific distribution is considered). The received signal power is $P H l(r_i)$. In the case of \textit{no interference} (or at least not considerable) between devices, the signal reception at $o$ is only affected by noise, with constant power $N>0$ $[Watt]$. For each link, communications takes place within a frequency band of $W$ [$Hz$], assigned by the operating system to guarantee the interference-free link (as in OFDMA). The quality of coverage provided by the transmitter $x_i$ to the origin is described by the Signal-to-Noise-Ratio $\mathrm{SNR}(r_i)$, equal to
\begin{equation}
\label{snr}
\mathrm{SNR}(r_i)=\frac{PHr_{i}^{-\alpha}} {N}.
\end{equation}

The maximum transmission rate from $x_i$ to $o$ in $[bits/sec]$ is given by the Shannon formula
\begin{eqnarray}
\label{rate}
\mathcal{R}(r_i) & = & W\log_2\left(1+\mathrm{SNR}(r_i)\right).
\end{eqnarray}

\subsection{Content and its Popularity}

The receiver at $o$ demands a file (say audio or video) from a finite set of $F>1$ available ones. This set is called \textit{content catalogue} and is denoted by $\textit{$C$} :=\{c_1,c_2....,c_F\}$. Each content (file) $c_j$ is related to a popularity value, assumed \textit{constant and known} a priori. The file popularity can be understood as the frequency of a particular content being called from an infinite stream of requests, and its law $\left\{a_j\right\}$, $j=1,\ldots,F$, can be any probability mass function of $F$ objects, with
\begin{eqnarray}
\sum_{j=1}^{F} a_j = 1.
\end{eqnarray}

If we assume that content is indexed in decreasing order of popularity  $j=1,\ldots,F$, and is Zipf-like distributed \cite{BreslauINFO99}, we find the traffic case of the Independent Reference Model (IRM), %\cite{CoffmanBook73}, 
with $a_j=A^{-1}j^{-\gamma}$, and $A:= \sum_{j=1}^{F} j^{-\gamma}$ being the normalising constraint. The Zipf exponent $\gamma$ characterises the distribution and depends on the type of content. %Table \ref{Tab2}, provides reference values of the exponent for different types of content. 
The Zipf distribution is \textit{heavy-tailed}\footnote{The definition of a heavy-tailed distribution is ambiguous in the literature. It often refers to distributions that are heavier than the Exponential. In a stronger sense it refers to distributions that have certain moments infinite, i.e. the mean, or just the variance. We will apply here the first definition, and as a consequence, the \textit{log-normal} distribution, which has all moments finite, will be considered also as heavy-tailed.}, which means that it can give rise to extremely large values with non-negligible probability. When $\gamma<2$, and for $F\rightarrow \infty$ the Zipf distribution has infinite mean value. 

%\vspace{+0.2cm}
%\begin{table}[h!]
%\caption{}
% \centering
%\begin{tabular}{|c|c|c|}
% \hline
% Content Type & Popularity Exponent ($\gamma$)  \\ [1 ex]
% \hline
%  VoD     & $0.65 \leq \gamma \leq 1.2$   \\  [0.5 ex]
% \hline
% Websites & $0.64 \leq \gamma \leq 0.83$  \\  [0.5 ex]
% \hline
% P2P Files & $0.75 \leq \gamma \leq 0.82$  \\  [0.5 ex]
% \hline
%\end{tabular}
%\label{Tab2}
%\end{table}

\subsection{Content Size}

Furthermore, we consider that each content $c_j\in C$ has a positive size \textit{$z_j>0$} given in [$bits$]. The size generally varies among different files, it is \textit{constant and known}, and depends on the content type. %Another important aspect is the relation of file's size with the transmission rate and time connexion for each mobile. In fact, the amount of time it takes receiver to download it file request will vary according to the file's size and Internet connexion speed.
A realisation of sizes for $F$ files can be \textit{sampled} from a probability distribution. There are many related studies in the literature that propose and make use of the distribution for the file sizes. Certain authors suggest that this distribution should be \textit{heavy-tailed} as well, as for the popularities, the reason being that traffic is dominated by multimedia content that is allowed to have very long duration. We found and propose here the following possibilities.
%\begin{itemize}

(A) Crovella and Bestavros \cite{CroBeSELF97} attribute a \textit{Pareto} distribution with parameter $0.9\leq a \leq 1.1$, which is verified to fit well with available network data sets. Its tail probability is
\begin{eqnarray}
\label{Pareto}
\bar{F}_{P}(z): = P\left[Z > z\right] = (\beta/z)^{a},
\end{eqnarray}
with  \textit{shape parameter} $a$ and \textit{scale parameter} $\beta > 0$, and $z\geq \beta$. The Pareto distribution has infinite variance for $a \leq 2$, and infinite mean for $a\leq 1$.

(B) Lee et al in \cite{NetWIFI13}, and Abhari and Soraya \cite{AbhariYouTube10} analysed measurements from YouTube videos and proposed a \textit{Weibull} distribution for video files, with tail
\begin{eqnarray}
\label{Weibull}
\bar{F}_{W}(z) = \exp(-(z/\mu)^{k}).
\end{eqnarray}
In general, for \textit{shape} $k<1$ the Weibull distribution exhibits infinity of certain first moments, and it becomes heavier as $k$ gets smaller. In these references, audio files are modelled by the \textit{Exponential} distribution (Weibull with $k=1$). 

(C) Certain authors, such as Downey in \cite{Dow01}, give evidence that the file-size distribution in the WWW, rather than Pareto, is actually \textit{Log-normal}; $\ln\mathcal{N}\left(\mu,\sigma\right)$. A random variable is log-normally distributed, if its natural logarithm follows the Normal distribution $\mathcal{N(\mu,\sigma)}$. An easy way to describe it is through its p.d.f.
\begin{eqnarray}
\label{lognormalpdf}
f_{\log}(z) = \frac{1}{z\sigma \sqrt{2\pi}}\exp\left(\frac{(\ln z - \mu)^2}{2\sigma^2}\right).
\end{eqnarray}

%\end{itemize}

The aforementioned publications, are not specific for mobile traffic - where users differ from wired networks in behavioural patterns and service needs. Mobile users are expected to show less interest for very large files, due to limited bandwidth of wireless links, storage limitations, and shorter mobile sessions especially outside home and work environments (see also \cite{Arv12}). For these reasons, we can propose to limit the file size by upper $z_{max}$ (also lower $z_{min}$) bounds, depending on the service type. This can be achieved by a \textit{truncated} distribution, e.g. truncated Log-normal, or simply by use of a \textit{Uniform} one. Video files are often larger (even of an order of magnitude) than audio files. A reasonable size range for these two media types related to mobile offloading is proposed in Table \ref{Tab1}.

\begin{table}[h!]
\caption{Uniform distribution}
\vspace{+0.5cm}
 \centering
 \normalsize
\begin{tabular}{|c|c|}
 \hline
 Content Type & Size Range (z)\\ [1 ex]% & Service Rate (Kbps) \\ [1 ex]
 \hline
  Audio file & 100 Kb - 20 Mb \\ [0.5 ex]% & 512  \\  [0.5 ex]
 \hline
 Video file  & 50 Mb - 2 Gb \\ [0.5 ex]%& 8000  \\  [0.5 ex]
 \hline
\end{tabular}
\label{Tab1}
\end{table}

%\subsection{content popularity law (case for Zipf Distribution)}

Altogether, a tuple of a-priori known values $(a_j,z_j)$ is related to each content  $c_j$. It is important to note that in \cite{BreslauINFO99} no strong correlation between file size and popularity was observed, except from the average size of popular contents being slightly smaller than that of the unpopular ones. 

\subsection{Content Placement to Caches}

The storage inventory of a D2D transmitter $x_{i}\in\Phi_t$ is denoted by $\Xi^{(i)}$ and contains a number of $|\Xi^{(i)}|\leq K$ distinct entire files from the catalogue \textit{$C$}. We consider that the content is independently installed in the caches by some probability distribution which guarantees that
	\begin{eqnarray}
	\label{bj}
    		b_j = Pr(c_j \in \Xi), & & 0\leq b_j \leq 1 , \ \forall {j},
     	\end{eqnarray}
i.e the probability (consequently frequency) that content \textit{$c_j$} is stored in any of the memory caches of the network devices is \textit{$b_j$}. In the above, due to independence, the superscript on inventories $(i)$ is dropped. These content placement probabilities are (pre)determined by the \textit{content placement policy}.

Since each transmitter $x_i$ has a memory of size $K$ [$objects$] irrespective of their file size, no more than $K$ distinct objects should be made available in each cache. In \cite{BlaGioICC15}, a probabilistic block placement (PBP) policy was suggested, that satisfies the sum constraint 
\begin{equation}
\label{probbjK}
\sum_{j=1}^{F} b_j \leq K,
\end{equation}
and at the same time guarantees the hard constraint $|\Xi^{(i)}| \leq K$, for all $i$. This is the policy we consider here as well.

%----------
\subsection{Device Mobility}

Node mobility is modelled in a simplified way as follows. At instant $t=0$, the receiver at the Cartesian origin sees a set of transmitters with fixed positions on the plane. The distance between the receiver and each transmitter varies over time, as a function $r_i(t)$, $\forall i$. Then, $\mathrm{SNR}(r_i(0))$ in (\ref{snr}) gives the signal quality and $\mathcal{R}(r_i(0))$ in (\ref{rate}) the achievable rate for the link between $o$ and $x_i$ at time origin. 

We need to include in the model that the position of each node, due to mobility, is bound to change. To do so, each link is considered active for a time period of $\tau_i$ [$sec$], different for each $x_i$, with $\mathrm{SNR}(r_i(t))=\mathrm{SNR}(r_i(0))$, within $0\leq t < \tau_i$. This is called its \textit{lifespan}. At $\tau_i$ the link is immediately dropped, because the transmitter instantly moves far away from its position so that $\mathrm{SNR}(r_i(\tau_i))\approx 0$. Of course, such a model with sudden node displacement does not describe the full complexity of random user movements and their trajectories on the plane, it is however sufficient in its simplicity to capture the main effects of mobility and to help better manage the cache inventories. 

The time intervals $\tau_i$ are random i.i.d. variables, that are also independent of all other parameters of the model, such as node position or fading. They can be seen as a mark of the process $\Phi_t$. The generic random variable for all $\tau_i$ is $T\sim Exponential(\bar{\tau}^{-1})$, with mean $\mathbb{E}\left[T\right] = \bar{\tau}$.

As a consequence of the model, the number of possible active links gradually reduces over time.

%-----------------------------------------------------

\section{Performance Metrics}
\label{secIII}

The user at the origin requests for a certain content from the catalogue $C$, which follows the popularity law $\left\{a_j\right\}$. We remind the reader that the placement probabilities on caches follow $\{b_j\}$. For a given content $c_j$, its service is successful if there exists \textit{at least} one transmitter $x_i\in\Phi_t$ who satisfies both conditions:

\emph{(i)} the object is found in its inventory, i.e. the event is true
\begin{equation}
\label{Aij}
A_{ij}:=\{c_j \in \Xi^{(i)}\},
\end{equation}

\emph{(ii)} the communication link between $o$ and $x_i$ is sufficient for the entire file to be transmitted, before the transmitter leaves its position and becomes unavailable. This can be formally expressed as the event
\begin{equation}
\label{Bij}
B_{ij}:={\left\{\tau_i W \log_2(1+\mathrm{SNR}(r_i)) \geq z_j\right\}},
\end{equation}
i.e. that the amount of information in [$bits$] transmitted over the link between $x_i$ and $o$ within the lifespan $\tau_i$ is greater than or equal to the total size $z_j$ of the object. Alternatively, we can write $\left\{\mathcal{R}(r_i)\geq z_j/\tau_i\right\}$, which poses a requirement for the throughput, to be greater than or equal to a threshold of minimum rate (file-size over transmitter's lifespan).

As mentioned, uninterrupted service is guaranteed when at least one transmitter satisfies (i) and (ii), but the case that more than one may exist with these conditions is not excluded. If so, we assume a random choice among these possibilities for a communication link with $o$ to be established, and also the existence of a schedule to achieve this. Our work mainly focuses on the probability that the service is satisfied. It does not consider issues on scheduling, resource allocation or load balancing (the load here being the receivers), when many users compete for access to the same transmitter. It neither considers the possibility for transmitters to cooperate for service. This option can be left open for future investigations. 

Having said this, service is achieved from $x_i$ to $o$ for object $c_j$, when 
\begin{equation}
   \Psi_{ij}:=\mathbf{1}_{\{A_{ij}\}} \mathbf{1}_{\{B_{ij}\}}=1,
\end{equation}
where $\mathbf{1}_{\{E\}}$ is the indicator function, which yields one, if the event $E$ is true, otherwise zero. We say that the receiver's demand is served successfully, when at least one transmitter $x_i\in\Phi_t$ exists, that satisfies $\Psi_{ij}=1$, i.e. if the following event is true
\begin{equation}
\label{Sj}
S_j := \{\sum_{i=1}^{\infty} \Psi_{ij} \geq 1\}.
\end{equation}

By denoting the probability of service for content $c_j\in C$ by $P_{srv,j}(b_j;z_j)$, we have
\begin{eqnarray}
\label{Pservj}
P_{srv,j}(b_j;z_j) = Pr(S_j) = Pr\left[\sum_{i=1}^{\infty} \Psi_{ij} \geq 1\right].
\end{eqnarray}

The performance metric is the \textit{Total Service Success Probability}, denoted by $P_{srv}$, and it is equal to the expected service probability over the content popularity, given the file sizes,
\begin{equation}
\label{Psrv}
 P_{srv}\left(\{b_j\};\left\{a_j\right\},\left\{z_j\right\}\right)= \sum_{j=1}^{F} a_j P_{srv,j}(b_j;z_j).
\end{equation}

Furthermore, we can take the distribution of file-size into consideration, by taking expectation of the above metric over the file-size variables. These are i.i.d. with c.d.f $F_Z(z)$ for each file size $z_j$, $j=1,\ldots,F$. We define the \textit{Expected Service Success Probability}, denoted by $\tilde{P}_{srv}(\{b_j\};\{a_j\})$, as

\begin{eqnarray}
\label{expPsrv}
\tilde{P}_{srv}(\{b_j\};\{a_j\}) & := & \mathbb{E}\left[P_{srv}\left(\left\{b_j\right\};\left\{a_j\right\},\left\{z_j\right\}\right)\right]\nonumber\\
& \stackrel{(\ref{Psrv}),indep.}{=} & \sum_{j=1}^{F} a_j \mathbb{E}\left[P_{srv,j}(b_j;Z)\right].
\end{eqnarray}

%--------------------------------

\section{Main Results}
\label{secIV}

\subsection{Per-Object Service Success Probability}

\begin{Pro} 
\label{Prop1}
The probability that the typical receiver requesting for object $c_j \in C$ is served by at least one transmitter in the downlink, for the $\mathrm{SNR}$ coverage model, is equal to
\begin{eqnarray}
\label{Psrvj}
P_{srv,j}(b_j;z_j) = 1 - \exp\left(-\pi \lambda_t b_j  \frac{I_H}{N^{2/ \alpha}} I_{T}(z_j,\bar{\tau}) \right),
\end{eqnarray}
where 
\begin{eqnarray}
\label{Eh}
I_H & := & P\mathbb{E}[H^{2/ \alpha}],\\
\label{Et}
I_{T}(z_j,\bar{\tau}) & := & \mathbb{E}[(\frac {1} {2^{\frac{z_j}{W T}}-1})^{2/ \alpha}].
\end{eqnarray}
\end{Pro}

\begin{proof}
We need to calculate (\ref{Pservj}). The randomness of the event (\ref{Sj}) is due to \emph{(a)} the position of the transmitters, which in our model follows a PPP, \emph{(b)} the lifespan $\tau_i\sim T$ of each node, \emph{(c)} the channel fading $h_i\sim H$ between each transmitter $x_i$ and the the origin $o$. Hence,
\begin{eqnarray}
\label{eqPsrv}
P_{srv,j} & = & Pr\left[\sum_{i = 1}^{\infty}\Psi_{ij}\geq 1\right] =  1-Pr\left[\sum_{i=1}^{\infty} \Psi_{ij}=0\right] \nonumber \\
& = & 1-Pr\left[\bigcap_{i=1}^{\infty} \{\Psi_{ij}=0\}\right]\nonumber \\
%& = & 1 - \mathbb{E}\left[\mathbf{1}_{\{\bigcap_{i=1}^{\infty} \{\Psi_{ij}=0\}\}}\right]\nonumber \\
& \overset{(a)}{=} & 1-\mathbb{E}\left[\prod_{i=1}^{\infty}\mathbf{1}_{\{\Psi_{ij}=0\}}\right]\nonumber \\
& \overset{(b)}{=} & 1-\mathbb{E}\left[\prod_{i=1}^{\infty}Pr(\Psi_{ij}=0)\right]\nonumber \\
& \overset{(c)}{=}  & 1-\exp\left(-\int\limits_{\reels^{2}}(1-u_{j} (x))\lambda_t dx\right)\nonumber \\
& \overset{(d)}{=}  & 1-\exp\left(-2\pi \lambda_t \int\limits_{0}^{\infty} (1-u_{j} (r))r dr\right).
\end{eqnarray}

In the above, (a) is due to the independence of the events $\Psi_{ij}$, (b) results by taking expectation over the i.i.d. variables $\{\tau_i\}$, $\{h_i\}$, (c) is obtained from the probability generating functional (PGFL) of the PPP, which states that for some function $u(x)$ it holds $\mathbb{E}[\prod_{x\in\Phi}{u(x)}]=exp{(-\lambda
\int_{\reels^{2}}(1-u(x))dx)}$ \cite{HaenggiBook13}, \cite{BacBlaVol1}. In our case, $u_{j} (x):= Pr(\Psi_{ij}=0)$ and $\lambda_t$ is the density of D2D transmitters. The last step (d) comes from Cartesian-to-Polar transformation.

To derive an explicit expression for the per-object Service Success Probability $P_{srv,j}$, the function $u_{j} (x_i)$ needs to be calculated.
 \begin{eqnarray}
 \label{eqn:15}
 u_{j} (x_i) & = & Pr(\Psi_{ij}=0) =  Pr(\mathbf{1}_{A_{ij}}\mathbf{1}_{B_{ij}}=0)\nonumber \\
 & = & Pr(\mathbf{1}_{\{A_{ij}\cap B_{ij}\}}=0)\nonumber \\
 %& \overset{(d)}{=} & P\left(({A_{ij}}\cap{B_{ij}})^c\right)\nonumber \\
 & = & 1-Pr({A_{ij}}\cap {B_{ij}})\nonumber \\
 & \overset{(e)}{=} & 1-Pr(A_{ij})Pr(B_{ij})\nonumber \\
 & := & 1-f_j(|x_i|).
\end{eqnarray}

In the above, (e) comes from the independence of the two events $A_{ij}$ and $B_{ij}$. We have replaced $f_j(|x_i|): = P(A_{ij})P(B_{ij})$, equal to $f_j(r_i)$, so that the function $u_j(x_i)$, depends only on the radial distance from $o$ to the transmitter. The reason is that,
\begin{eqnarray}
\label{fj}
f_j(r_i) & := & \underset{(i)}{\underbrace{Pr(c_j \in \Xi_i)}}\underset{(ii)}{\underbrace{Pr(\tau_i W\log_2 (1+\mathrm{SNR}(r_i)) \geq z_j)}}\nonumber\\
 & \overset{\tau_i\sim T}{=}  & b_j Pr\left(\mathrm{SNR}(r_i) \geq 2^{\frac{z_j}{W T}}-1\right).
\end{eqnarray}
%\begin{equation}
% \begin{split}
%\begin{align}
%Pr(\tau_i \ln (1+SNR(r_i)) \geq z_j) & = Pr(SNR(r_i)) \geq e^{z_j /\tau_i}-1)\\
%& = Pr(\frac{h_i} {r_{i}^{\alpha}\sigma^{2}} \geq e^{ z_j / \tau_i}-1)\nonumber \\
%& \overset{(f)}{=} g_j(r_i)\label{eqn:16}
%\end{align}
% \end{split}
% \end{equation}
%Where (f) is obtained by using the fact that $h_i$ and $\tau_i$ follow some distribution which will define the probability.
Substituting (\ref{eqn:15}) into (\ref{eqPsrv}) yields the following
 \begin{eqnarray}
 P_{srv,j} %& = 1-\exp{(-2\pi \lambda \int\limits_{0}^{\infty} (1-(1-f_j(r)))r dr)}\nonumber \\
 & = &  1-\exp\left(-2\pi \lambda_t \int\limits_{0}^{\infty} f_j(r)r dr\right).
 \label{eqn:20}
\end{eqnarray}

Using also the equality in (\ref{fj}) and the $\mathrm{SNR}$ definition in (\ref{snr}), the expression becomes
\begin{eqnarray}
P_{srv,j} & \overset{(k)}{=}  & 1-\exp\left(-2\pi \lambda_t b_j \int\limits_{0}^{\infty} Pr(r^{\alpha} \leq \tilde{R}_j)rdr\right)\nonumber\\
& = & 1-\exp\left(-2\pi \lambda_t b_j \int\limits_{0}^{\infty}\mathbb{E}[\mathbf{1}_{\{r^{\alpha} \leq \tilde{R}_j\}}]rdr\right)\nonumber\\
&  \overset{(\ell)}{=}  & 1-\exp\left(-2\pi \lambda_t b_j \mathbb{E} \left[\int\limits_{0}^{\infty} \mathbf{1}_{\{ r \leq \tilde{R}_j^{1/\alpha} \}} rdr\right]\right)\nonumber\\
& = & 1-\exp\left(-2\pi \lambda_t b_j \mathbb{E} \left[\int_{0}^{\tilde{R}_j^{1/\alpha}}rdr\right]\right)\nonumber\\
& = & 1-\exp\left(-2\pi \lambda_t b_j \mathbb{E} \left[\frac {r^{2}} {2} |_{0}^{\tilde{R}_j^{1/\alpha}}\right]\right)\nonumber\\
& = & 1-\exp\left(-\pi \lambda_t b_j \mathbb{E}[\tilde{R}_j^{2/ \alpha}]\right),\nonumber
 \end{eqnarray}
 where in (k) we substitute $\tilde{R}_j := \frac {PH}{N} \frac {1}{2^{\frac{z_j}{WT}}-1}$. Hence $\mathbb{E}[\tilde{R}^{2/ \alpha}] = \frac{P\mathbb{E}[H^{2/\alpha}]}{N^{2/\alpha}}\mathbb{E}[ \left(2^{\frac{z_j}{WT}}-1\right)^{-2/\alpha}]$. The order of integration is interchanged in ($\ell$) due to Fubini's theorem. 
\end{proof}

\textit{Discussion on Proposition 1:} From Eq. (\ref{Psrvj}) we observe that the service probability for object $c_j$ is a function of its content placement probability $b_j$ and its file-size $z_j$. In fact, the function $P_{srv,j}$ is increasing in $b_j$. Regarding the file-size, the function is decreasing in $z_j$. Hence, service exhibits the behaviour that one would expect related to object characteristics. Another dependence of the function is on the average lifespan $\bar{\tau}$, and it is increasing as the mean lifespan increases. Actually, a larger lifespan corresponds in our model to \textit{low mobility}.

Furthermore, the service probability depends also on system parameters, such as the path-loss exponent $\alpha$, the transmitter density $\lambda_t$, their transmit power $P$ and bandwidth $W$. As $\lambda_t$ increases, and for $b_j>0$, the $P_{srv,j}$ tends to 1. The reason for such behaviour is the $\mathrm{SNR}$ communications model, which does not consider interference, so densifying the network, simply guarantees that the receiver will find its request somewhere close with sufficiently high probability. The function is further increasing in $P$ and $W$ and decreasing in $\alpha$, the latter because the coverage area of each transmitter decreases for increasing path-loss exponent.

The expression $P_{srv,j}$ depends on the expectations $I_H$  and $I_T$ in (\ref{Eh}) and (\ref{Et}) respectively. 

\subsubsection{Fading distribution for $I_H$}

To obtain specific expressions for $\esperance[H^{2/ \alpha}]$ we need to determine the type of distribution for the propagation effects $H$ (shadowing, and/or fading) experienced by the typical receiver. A list of such distributions with their density function and the calculation of $I_H$ is found in \cite{BartekPIMRC13}. We reproduce this here in Table \ref{Tab2}. As an example, if we assume exponential fading with mean 1 (as is often the case), the expression for the service probability gives
  \begin{eqnarray}
  P_{srv,j}(b_j;z_j) = 1-\exp(-\pi \lambda_t b_j  \frac{P^{2/\alpha}\Gamma(\frac{2}{\alpha}+1)}{N^{2/\alpha}} I_T(z_j,\bar{\tau})).\nonumber
  \end{eqnarray}

%-------------------------------------------------------------------

\begin{table}[!t] 
 \centering
 \caption{Fading distributions, taken from \cite{BartekPIMRC13}.}
 \vspace{+0.5cm}
 \begin{tabular}{|c|c|c|} 
 \hline
 Distribution & Probability density of H & $\esperance[H^{2/ \alpha}]$\\ 
 \hline
 & $1/(h\sqrt{2\pi}\sigma)$& \\
 Log-normal & $\times e^{-(\ln{h}-\mu)^2/2\sigma^2}$ & $e^{2(\sigma^2+\mu \alpha)/\alpha^2}$\\
%[1 ex]
 \hline
 Exponential & $\lambda e^{-\lambda h}$ & $\lambda^{-2/\alpha} \Gamma(2/ \alpha+1)$\\ % [1 ex]
 \hline
  & $(k/\lambda)(h/\lambda)^{k-1}$&\\
  Weibull & $\times e^{(-h/\lambda)^k}$ & $\lambda^{2/\alpha} \Gamma(2/(\alpha k)+1)$  \\ % [1ex]
 \hline
 & $2m^m/ \Gamma(m) \Omega^m h^{2m-1}$&\\
 Nakagami & $\times e^{-(m/\Omega)h^2}$ &  $\Gamma(1/{\alpha+m})\Gamma(m) (\Omega/m)^{1/\alpha}$ \\ %[1ex]
 \hline
  & $(h/\sigma^2)I_{0}(h\nu/\sigma^2)$ &$(2\sigma^2)^{(1/\alpha)} \Gamma(1/\alpha+1)_\mathbf{1}$\\
  Rice & $e^{-(h^2+\nu^2)/(2\sigma^2)}$ & $ \times F_{1}(-1/\alpha,1;-\nu^2/(2\sigma^2)$  \\% [1ex]
 \hline
\end{tabular}
\label{Tab2}
\end{table}

%--------------------------------------------------

\subsubsection{Lifespan distribution for $I_T$}

The distribution for node lifespan $T$ plays an important role for performance. To get first intuition we can use (i) a fixed period $\tau_i=\bar{\tau}$, $\forall x_i$. In this case, all nodes will be available for data transmission during the period $\left[0,\bar{\tau}\right]$ and after that, no service will be provided. Since lifespan does not vary, the probability of service mostly depends on the probability a transmitter with the desired object to be sufficiently close to the origin. This distribution simplifies the calculations, giving
\begin{eqnarray}
\label{ITfix}
I_T(z_j,\bar{\tau}) \overset{T=\bar{\tau}}{=} \left(2^{\frac{z_j}{W\bar{\tau}}}-1\right)^{-2/\alpha}.
\end{eqnarray}

A more realistic assumption is that of an exponential distribution. Each node has the ability to provide service for an exponential time while keeping its position, before moving away from the origin. In this case, where $T\sim Exponential(\bar{\tau}^{-1})$, the factor $I_T$ can be calculated by the integral 
 \begin{eqnarray}
 \label{ITexp}
 I_T(z_j,\bar{\tau}) \overset{Exp}{=} \int_{0}^{\infty} \frac{1}{\bar{\tau}} \exp(-\frac{s}{\bar{\tau}}) \left(2^{\frac{z_j}{Ws}}-1\right)^{-2/ \alpha} ds.
\end{eqnarray}

\subsection{Total and Expected Service Success Probability}
Using Prop. \ref{Prop1}, and applying this to the expression in (\ref{Psrv}) for the total success metric, we get the following result. 

\begin{Cor}
\label{Cor1}
The Total Service Success Probability for the $\mathrm{SNR}$ coverage model is equal to (we omit the dependence on $\left(\{b_j\};\{a_j\},\{z_j\}\right)$ due to space limitation),
\begin{eqnarray}
\label{TotCor1}
P_{srv} = 1-\sum_{j=1}^F a_j\exp\left(-\pi\lambda_t b_j \frac{I_H}{N^{2/\alpha}} I_T(z_j,\bar{\tau})\right).
\end{eqnarray}
\end{Cor}

By further assuming independently sampled file-sizes from a distribution with c.d.f. $F_Z(z)$ and p.d.f. (for continuous functions) or p.m.f. (for discrete functions) $f_Z(z)$, the expected success metric  takes the following expression.

\begin{Cor}
\label{Cor2}
The Expected Service Success Probability for the $\mathrm{SNR}$ coverage model is equal to (we omit the dependence on $\left(\{b_j\};\{a_j\}\right)$ due to space limitation),
\begin{eqnarray}
\label{ETotCor1}
\tilde{P}_{srv} = 1-\sum_{j=1}^F a_j\mathbb{E}\left[\exp\left(-\pi\lambda_t b_j \frac{I_H}{N^{2/\alpha}} I_T(Z,\bar{\tau})\right)\right].
\end{eqnarray}
\end{Cor}

\section{Numerical Evaluation}
\label{secV}

In this section, the expressions derived for the performance of the D2D model under study are numerically evaluated. Additionally, we have run extended simulations of the system, to validate their correctness. 

Specifically for the simulation environment, we consider the following. \textit{Geometry Parameters:} The simulation is observed within a window of size $100\times 100$ [$km^2$], where transmitter devices are distributed as a Poisson point process (PPP) of density $\lambda_t = 2.5 \cdot 10^{-3}$ [$transmitters/m^2$]. With this transmitter density, the mean distance from $o$ to the closest transmitter is $r_{first} = (2\sqrt{\lambda_t})^{-1} = 10$ $[m]$. Each node has a lifespan that is exponentially distributed with mean value that varies for (a) Audio files $\bar{\tau} \in \left[0, 100\right]$ [$sec$], and (b) for Video files $\bar{\tau}\in\left[0, 1000\right]$ [$sec$]. The receivers form also a PPP but we consider just one receiver per realisation placed at the Cartesian origin $o$, since no resource sharing is assumed. The simulation results are averaged over $T_{sim} = 2000$ iterations (a larger number gives even better fit). \textit{Wireless Parameters:} The transmission is interference free. Each node operates on a bandwidth of $5$ [$MHz$] and emits with Power $P=0.5$ [$Watt/Hz$], whereas noise power is $N=10^{-11}$ [$Watt/Hz$]. The path-loss exponent is $\alpha=4$ and fading is Rayleigh (hence exponential distribution). \textit{Content Parameters:} We consider a catalogue of size $F=100$ objects, and a Zipf distribution for popularity, with exponent $\gamma=0.78$. The objects can be (a) Audio files, or (b) Video files. The cache size is $K=5$ [$objects$] and we do not consider the influence of file-sizes when filling in the node inventories.

\textit{Content Placement:} We use the probabilistic block placement policy (PBP) proposed in \cite{BlaGioICC15}, which for each node samples an independent vector of at most $K$ objects, and satisfies the placement probabilities $b_j\geq 0$ for every object $c_j$. The vector $\{b_j\}$ should be given as system input. The choice of entries is critical for the system performance itself and is a design parameter. In the current simulation we take $b_j >0$, for $1\leq j \leq 2K$ and $0$, otherwise. Then, for $1\leq j \leq 2K$, $b_j = a_j^*$, where $a_j^*$ is a sort of normalised popularity $a_j^* := \min\left\{\frac{Ka_j}{\sum_{k=1}^{2K} a_j}, 1\right\}$, so that $\sum_{j=1}^F b_j = \sum_{j=1}^{2K} b_j \leq K$.

We investigate different file-size distributions for the random variable $Z$. In all cases, the choice of file-size per object $c_j$ is i.i.d. and $z_j\sim Z$, $\forall c_j\in C$, so that the mean size for Audio is $\approx 10$ [$Mb$] and for Video $\approx 1$ [$Gb$]. In the case of a \textit{Uniform} distribution, the range of file-sizes for Audio and Video given in Table \ref{Tab1} satisfy the values of the expectation.

\subsection{Validation}

The correctness of the expression in (\ref{TotCor1}) is validated for the case of Audio and Video files separately. Both sets of file-sizes ($100$ in total in each set) are \textit{sampled} from an exponential distribution with mean values $10$ [$Mb$] (audio) and $1$ [$Gb$] (video) respectively. The expression for $I_T$ is given in (\ref{ITexp}) because lifespan distribution is exponential as well. The comparison between analysis and simulation is illustrated in Fig. \ref{fig:valid} for the Total Service Success Probability over a range of mean lifespan values, which is chosen differently for the two file categories. The plots show an excellent match between analysis and simulations. Interestingly, we observe that for $\bar{\tau}=100$ [$sec$], audio files have a success probability $P_{srv,Audio} \approx 0.37$, much higher than the success probability of videos for the same value of mean lifespan $P_{srv,Video} \approx 0.04$. This is reasonable due to the difference in mean file-size. Both sub-figures show a diminishing increase of $P_{srv}$. The probability should converge to some value less than one, because of the placement policy, which leaves $F-2K$ objects definitely uncached. Moving on the $x$-axis in both plots from left to right represents a change in the D2D behaviour, from higher to lower \textit{mobility}.
\begin{figure}[ht!]
\centering
\includegraphics[width=1.7in]{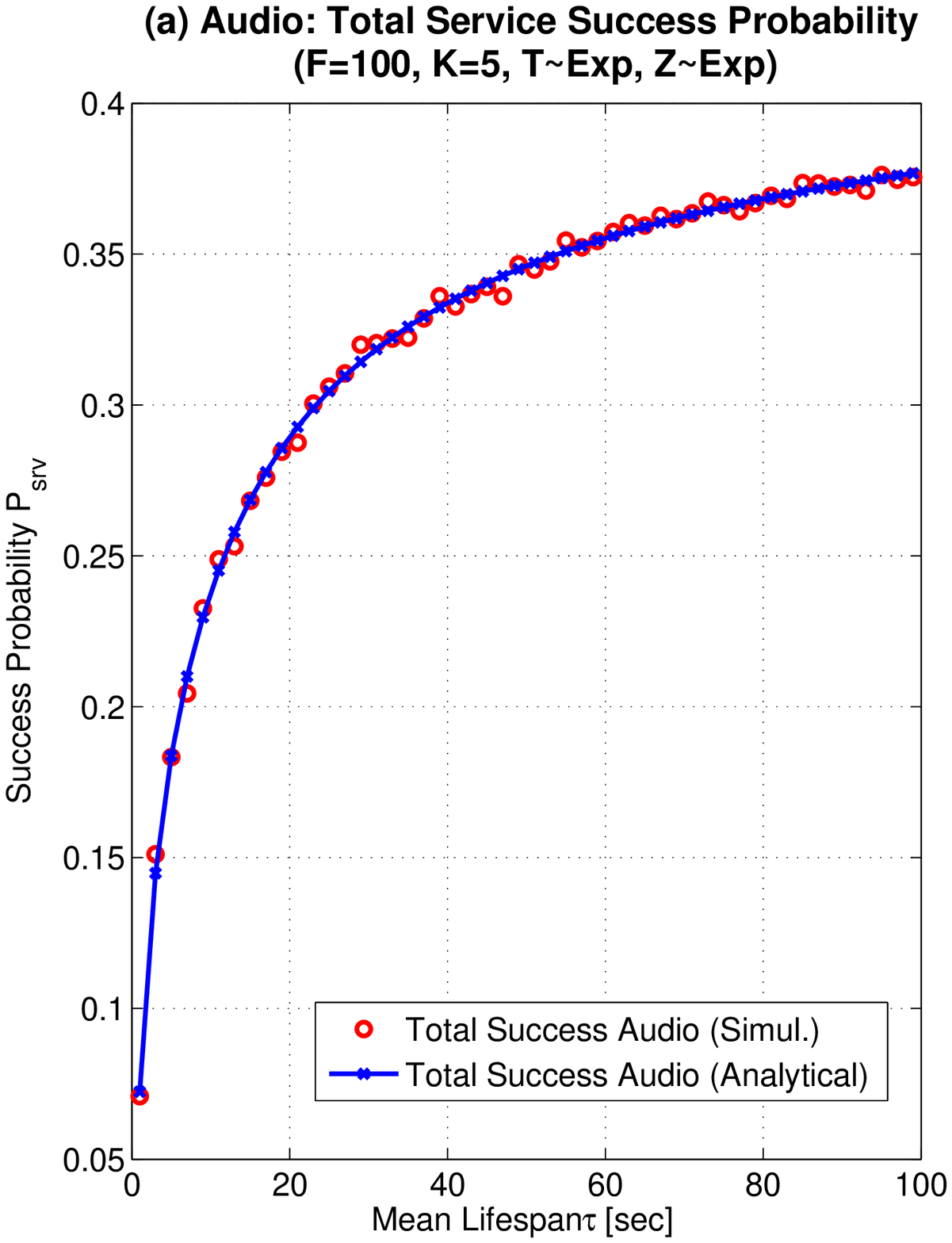}
\includegraphics[width=1.7in]{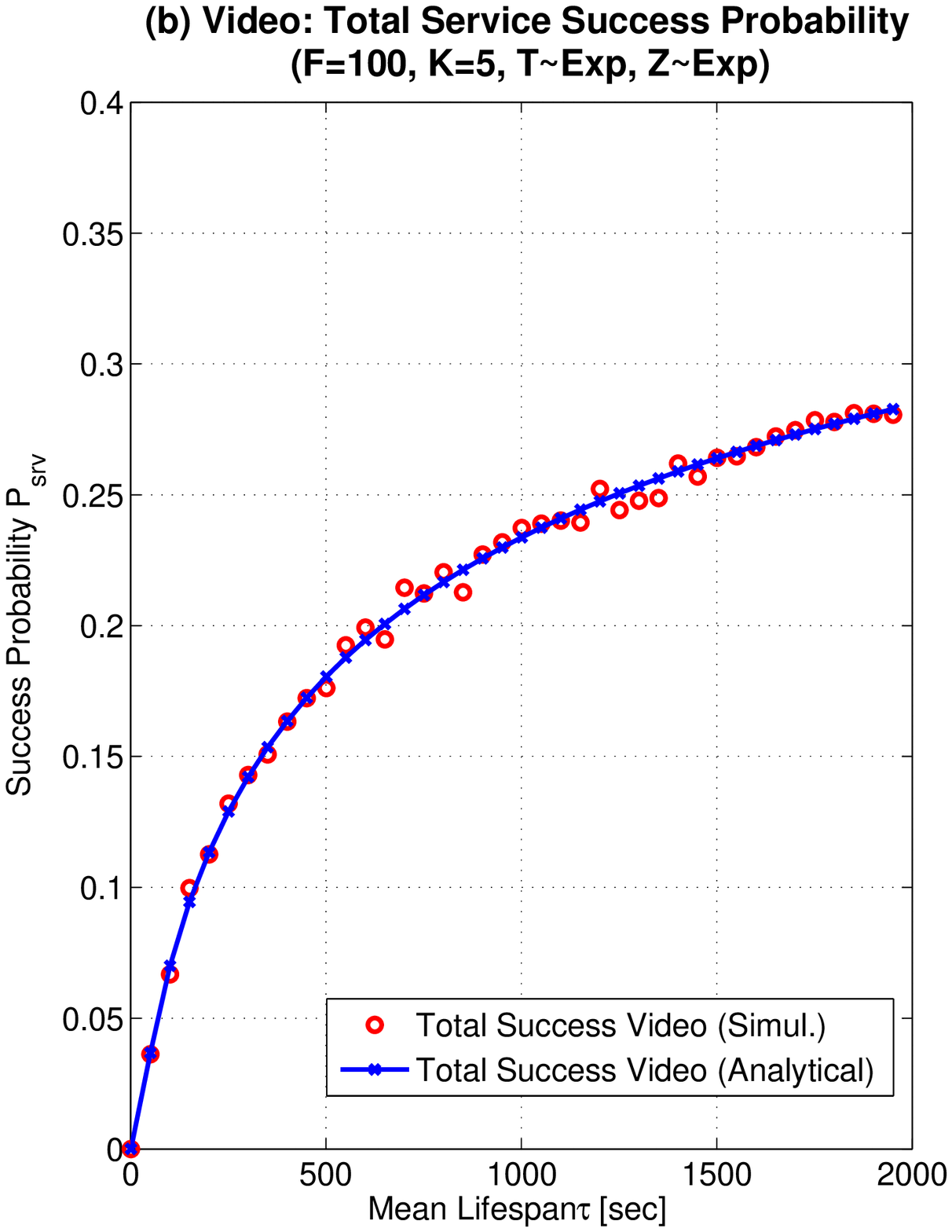}
\caption{Total service probability with respect to mean node lifespan, for cached (a) Audio files and (b) Video files.}
\label{fig:valid}
\end{figure}

\subsection{Evaluation}

\subsubsection{Influence of file-size order related to popularity}
The results in Fig. \ref{fig:valid} are obtained when the file-size of each content is independently sampled from an exponential distribution. There is, hence, no correlation between popularity and file-size, meaning that the most popular file may have any size. We investigate how the service success probability is influenced from a possible correlation between the two file characteristics. More specifically, we simulate (and analytically calculate) two scenarios, one when the file-sizes are in decreasing order in relation to the popularity index (the most popular file is the largest one from the sample set, the second most popular the second largest and so on), and another scenario when the file sizes are in increasing order (most popular file is the smallest one). These two curves over the mean lifespan are produced for Video in Fig. \ref{fig:Corr} (for Audio the behaviour is similar), and can be directly compared to those in Fig. \ref{fig:valid}(b). We observe that when the file-sizes are in increasing order, the $P_{srv}$ is higher than in the independent case, because more popular files (which are also cached due to the choice of the placement policy) are smaller and thus more likely to be fully transmitted within the lifespan. On the other hand, when the sizes are in decreasing order, the $P_{srv}$ is lower than in the independent case, for exactly the opposite reason.
\begin{figure}[ht!]
\centering
\includegraphics[width=3.4in]{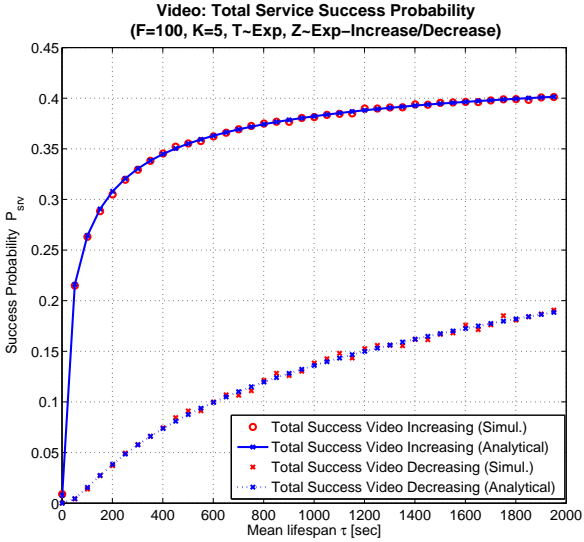}
\caption{Increasing/Decreasing file-size with popularity order and how it influences the total service probability, for cached Video files, over the mean node lifespan.}
\label{fig:Corr}
\end{figure}

\subsubsection{Influence of file-size distribution in the Expected/Total Service Success Probability}
We investigate this very important issue with the following methodology. We consider five possible distributions for the file-size of Videos: \textit{A}. \textit{Uniform} within $\left[z_{\min}, z_{\max}\right]$, \textit{B}. \textit{Exponential} with parameter $\lambda$, \textit{C}. \textit{Pareto} with parameters $(a,\beta)$, \textit{D}. \textit{Weibull} with parameters $(\mu,k)$, \textit{E}. \textit{Log-Normal} with parameters $(\mu,\sigma)$. Each one of these has very different characteristics and the system performance depends on the parameter values. 

To keep comparison fair, we choose the parameters so that $\mathbb{E}\left[Z\right] = 1$ [$Gb$], which is the mean value of the video file-size. Having this in mind, \textit{A}. the Uniform is in agreement with the expected value when using the upper and lower bounds in Table \ref{Tab1}. \textit{B}. For Exponential $\lambda = 10^{-9}$. \textit{C}. For Pareto it holds that $\mathbb{E}[Z] = \frac{\beta a}{a-1}$, and $a>1$ to have finite first moment. Then we choose $a=20/19$ and $\beta=0.05\cdot 10^{9}$ which not only gives the desired expected value but also guarantees the same lower bound $z\geq \beta=z_{\min}$  with the Uniform distribution. \textit{D}. For Weibull, $\mathbb{E}\left[Z\right]=\mu\Gamma(1+\frac{1}{k})$, so that the choice is $\mu=276$ and $k=0.1<1$. Finally, \textit{E}. for Log-normal, $\mathbb{E}\left[Z\right]=\exp(\mu+\frac{\sigma^2}{2})$ and we choose $\mu=5\ln(10)$, $\sigma=\sqrt{8\ln(10)}$.

To observe the influence of these file-size distributions on the probability of service, we first make use of the \textit{Expected Service Success Probability} metric, in (\ref{ETotCor1}). To simplify numerical evaluation, we use a fixed lifespan for all nodes $T=\bar{\tau}=1000$ $sec$, so that $I_T$ is given in (\ref{ITfix}). The results of the evaluations, with the parameters of each distribution chosen as above, are given in Fig. \ref{fig:Expdistvar}. From the plots, we observe that the distributions are ordered as: $A.\ Uni <B.\ Exp <C.\ Par <E.\ Log <D.\ Weib$. We note that the three heavy-tailed (or just heavier than the exponential) distributions Weibull, Pareto and Log-Normal, with the parameter set chosen, give the maximum service performance. The reason that the exponential performs poorly is that, although it does not have the tendency to generate very high values, it does however produce a sufficient number of samples large enough to keep the $P_{srv}$ low. Contrary to this, the heavy-tailed distributions may produce extremely large samples, however not a large number of them, so that small files tend to have smaller size than the ones from the exponential. Different heavy-tailed distributions can generate samples which deviate considerably from the mean towards higher values, but with even smaller low values, in order to keep the same $\mathbb{E}[Z] = 1$ [$Gb$]. The reason for the poor performance of the uniform is that, although bounded between $[z_{\min}, z_{\max}]=[0.05, 2]$ [$Gb$], a considerable amount of samples, due to uniform sampling, will be around the highest value, whereas no samples can be smaller than the lowest bound. The performance in the plots converges to an upper bound $\mathbb{E} [P_{srv} ]\approx 0.3433$, equal to the sum of popularities of the $2K$ cached most popular files $\sum_{j=1}^{2K}a_j$ for $\gamma=0.78$.

\begin{figure}[ht!]
\centering
\includegraphics[width=3.4in]{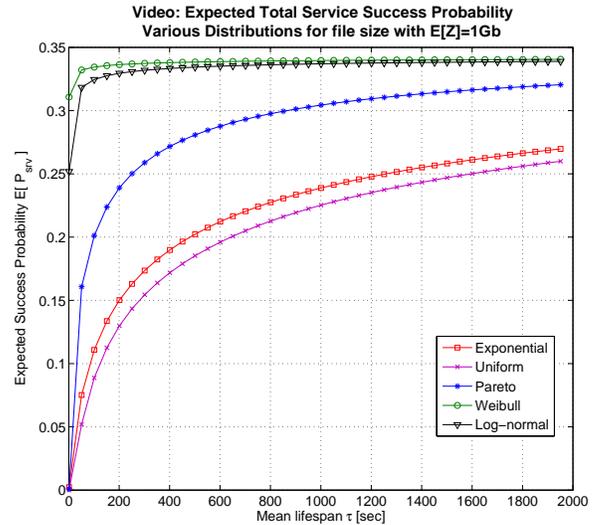}
\caption{Comparison of Expected Service Success Probability, with different file-size distributions 
for Video.}
\label{fig:Expdistvar}
\end{figure}

\vspace{+0.5cm}
To further give intuition on the effect of the tail distribution on the performance, we proceed in an alternative way: a larger sample set of size $F=200$ values is produced from each distribution. The samples of each set are then indexed in decreasing order. Table \ref{Tab3} shows the five highest values per set. Using these, the Total Service Success Probability per mean lifespan is plotted, which shows the same results as the expected case, with the difference in the performance of the Uniform distribution, which is higher due to the file-size upper bound of $2$ [$Gb$]. Both Fig. \ref{fig:Vardec} and Table \ref{Tab3} support the intuition and our explanations given above.

\begin{table}[ht!]
\normalsize
\caption{Samples of file-sizes [$Gb$] in descending order.}
\vspace{+0.5cm}
 \centering
\begin{tabular}{|l|c|c|c|c|c|}
 \hline
Distrib./Obj. & $c_1$ & $c_2$& $c_3$& $c_4$& $c_5$ \\ [1 ex]
 \hline
A. Uniform & 2.00 & 1.99 & 1.98 & 1.97 & 1.97 \\ [0.5 ex]
 \hline
B. Exponential & 6.21 & 5.64 & 5.33 & 4.22 & 4.01 \\ [0.5 ex]
\hline
C. Pareto & 8.93 & 6.37 & 5.70 & 4.92 & 1.47  \\ [0.5 ex]
\hline
D. Weibull & 24.00 & 18.07 & 1.45 & 0.07 & 0.04  \\ [0.5 ex]
\hline
E. Log-Normal & 7.28 & 7.05 & 2.49 & 1.13 & 1.05  \\ [0.5 ex]
\hline
\end{tabular}
\label{Tab3}
\end{table}

\begin{figure}[ht!]
\centering
\includegraphics[width=3.4in]{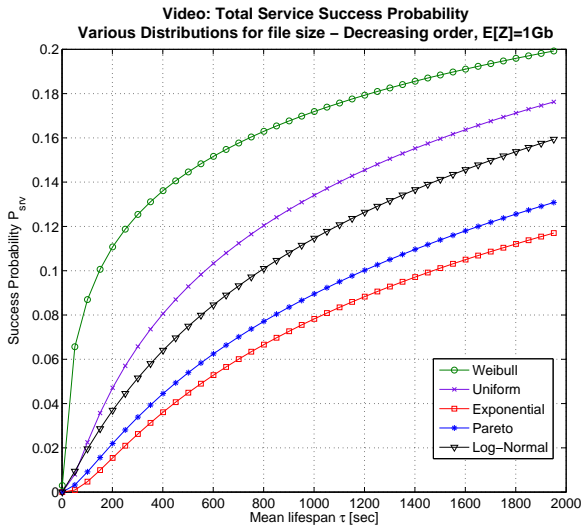}
\caption{Comparison of Total Service Success Probability, with different file-size distributions for a Video sample set of $F=200$ objects in decreasing order related to popularity.}
\label{fig:Vardec}
\end{figure}

\section{Conclusions}
\label{secVI}
The influence of mobility in the performance of cached D2D networks has been investigated. In the proposed model, a link is successful if the object request from the transmitter is cached at the memory of the receiver, and the link quality is sufficient for the entire object to be transferred before the transmitter leaves its place. The change in position is assumed sudden and the link is immediately dropped afterwards. This is a simplified approach to model mobility, and more realistic approaches can be considered in the future where the link quality could evolve more gradually. The work further investigates the influence of different file-size distributions on the network performance, and the analysis concludes that the exponential distribution may underestimate performance compared to heavy-tailed ones. Furthermore, it is shown that caching smaller files is in general beneficial, since these are more probable to be successfully transferred. Further extensions of the work may include the possibility of cooperative transmission \cite{QuekCoopCache16}, as well as an evaluation when interference influences the coverage probability. 

\section{Acknowledgement}
The authors would like to thank Bart{\l}omiej B{\l}aszczyszyn for the discussions and contributions throughout this research.

\bibliography{D2D}
\bibliographystyle{plain}

\end{document}